%% file: GC.tex
\documentclass[10pt,conference]{IEEEtran}

\usepackage[draft]{hyperref}  
\usepackage{gensymb}
\usepackage{amsmath}
\usepackage{amsfonts}
\usepackage{epsfig}
\usepackage{amssymb}
\usepackage{cite}
\usepackage{hhline}
\usepackage{multirow}
\usepackage{framed}
\usepackage{xcolor}
\usepackage{lipsum}
\usepackage{enumerate}
\usepackage{color}
\usepackage{graphicx}
\usepackage{subfigure}
\usepackage{hyperref}
\usepackage{algorithmic}
\usepackage[ruled]{algorithm}
\usepackage{booktabs}
\usepackage{bm}

\showoutput
\begin{document}

\title{Millimeter Wave Ad~Hoc Networks: \\Noise-limited or Interference-limited?\vskip -0.5cm}

\author{\IEEEauthorblockN{Hossein Shokri-Ghadikolaei and Carlo Fischione}
\IEEEauthorblockA{Automatic Control Department, Electrical Engineering and ACCESS \\
KTH Royal Institute of Technology, 10044, Stockholm, Sweden \\
{emails: \{hshokri, carlofi\}@kth.se}}}

\input{commands}
\maketitle

\begin{abstract}
In millimeter wave (mmWave) communication systems, narrow beam operations overcome severe channel attenuations, reduce multiuser interference, and thus introduce the new concept of noise-limited mmWave wireless networks. The regime of the network, whether noise-limited or interference-limited, heavily reflects on the medium access control (MAC) layer throughput and on proper resource allocation and interference management strategies. Yet, alternating presence of these regimes and, more importantly, their dependence on the mmWave design parameters are ignored in the current approaches to mmWave MAC layer design, with the potential disastrous consequences on the throughput/delay performance. In this paper, tractable closed-form expressions for collision probability and MAC layer throughput of mmWave networks, operating under slotted ALOHA and TDMA, are derived. The new analysis reveals that mmWave networks may exhibit a non-negligible transitional behavior from a noise-limited regime to an interference-limited regime, depending on the density of the transmitters, density and size of obstacles, transmission probability, beamwidth, and transmit power. It is concluded that a new framework of adaptive hybrid resource allocation procedure, containing a proactive contention-based phase followed by a reactive contention-free one with dynamic phase duration, is necessary to cope with such transitional behavior.
\end{abstract}


\section{Introduction}\label{sec: introductions}
Increased demands for higher data rates and limited available spectrum for wireless systems below 6~GHz motivate the use of millimeter wave (mmWave) communications to support multi-gigabit data rates. This interest has led to several standards for indoor wireless personal area networks (WPANs) or wireless local area networks (WLANs) such as IEEE~802.15.3c~\cite{802_15_3c}, IEEE~802.11ad~\cite{802_11ad}.

A mmWave communication has short wavelength, large bandwidth, and high attenuation through most obstacles~\cite{Rangan2014Millimeter,Rappaport2015wideband,shokri2015mmWavecellular}. Very small wavelengths allow the implementation of a large number of antenna elements in the current size of radio chips, which promises substantial antenna gains using narrow beams both at the transmitter and at the receiver. Besides boosting the link budget, these pencil beams will reduce the multiuser interference in the network. In the extreme case, once multiuser interference does no longer limit the network throughput, we face a noise-limited network where the achievable throughput is limited only by the noise power at the receiver. An interesting question is whether a mmWave network with pencil-beam operation is noise-limited, as opposed to conventional interference-limited wireless networks.
The answer of this question reveals the required complexity of different medium access control (MAC) layer functions. As the system goes to the noise-limited regime, design of proper resource allocation and interference management mechanisms is substantially simplified~\cite{shokri2015mmWavecellular,Shokri2015mmWaveWPAN,Singh2011Interference,Shokri2015Beam}. For instance, a very simple resource allocation such as activating all links at the same time without any coordinations among different links may outperform a complicated independent-set based resource allocation in a noise-limited regime~\cite{Shokri2015Beam}. Instead, pencil-beam operation complicates negotiation among different devices in a network, as control message exchange may require a time consuming alignment procedure~\cite{Shokri2015Beam}. That is, the time required to find the best set of beams at the transmitter and at the receiver to establish a mmWave link.

The seminal work of Singh \emph{et al.}~\cite{Singh2011Interference} shows the validity of \emph{pseudowired} abstraction (noise-limited network) in outdoor mmWave mesh networks. However, as shown in~\cite{Qiao2012STDMA,Shokri2015Beam}, indoor mmWave WPANs are not necessarily noise-limited. In particular, activating all links may cause a significant performance drop compared to the optimal resource allocation~\cite{Shokri2015Beam}, indicating that there may be situations in which a non-negligible multiuser interference is present; the noise power is not the only limiting factor.
Such a performance degradation increases with the number of devices in the network~\cite{Shokri2015Beam}. This indeed means that the accuracy of the noise-limited assumption to model the actual network behavior reduces with the number of links. Similar conclusions are also made in the context of mmWave cellular networks in~\cite{di2014stochastic}. The increased directionality level in a mmWave network reduces multiuser interference; however, this reduction may not be enough to take an action (e.g., resource allocation) based on the assumption of being in a noise-limited regime. It follows that a pseudowired assumption may be detrimental for the MAC layer design.

In this paper, we analyze the MAC layer throughput of a mmWave ad~hoc network, operating under slotted ALOHA protocol. To this end, we introduce a novel blockage model that captures the correlation among line-of-sight (LoS) events of different links. We derive tractable closed-form expressions for collision probability, per-link throughput, and area spectral efficiency. We analytically evaluate the impact of the transmission/reception beamwidth, transmission power, and the densities of the transmitters and obstacles on the performance metrics. The new analysis establishes that the pseudowired abstraction may not be accurate even for a modest-size mmWave ad~hoc network, and, more importantly, that mmWave networks exhibit a \emph{transitional behavior}, from a noise-limited regime to an interference-limited one. Specifically, we investigate when the interference footprint of the network does not show a binary behavior, and we show for which system parameters the network exhibits a transitional behavior during which interference may have degrees of severity. In the presence of such behavior, we investigate the pros and cons of collision-free and collision-based resource allocation protocols in mmWave systems. We conclude that a simple ALOHA protocol may significantly outperform both per-link throughput and area spectral efficiency of time division multiple access (TDMA) protocol with much less signaling overhead, while TDMA is still necessary to guarantee communication without any collision.
Detailed discussions of this paper provide guidelines for designing efficient resource allocation in mmWave networks with transitional behaviors.

The rest of the paper is organized as follows. In Section~\ref{sec: system model}, we describe the system model. The analysis of collision probability and throughput are provided in Section~\ref{sec: performance-evaluation}. Numerical results is presented in Section~\ref{sec: numerical-results}, followed by conclusions and future works in Section~\ref{sec: Conclusion}.

\section{System Model}\label{sec: system model}
We consider a mmWave wireless network with a homogeneous Poisson network of transmitters on the plane with density $\lambda_t$ per unit area, each associated to a receiver. To evaluate the collision performance of the network, we consider a reference link (called typical link) between a typical receiver, located at the origin of the Polar coordinates, and its intended transmitter, located at distance $L$ from the origin. We assume that if multiple neighbors are transmitting to the same receiver, at most one of them can be successfully decoded by that receiver~\cite{Singh2011Interference}.
Therefore, all transmitters in the network act as potential interferers for the typical receiver (the receiver of the typical link).

We consider a slotted ALOHA protocol without power control to derive a lower bound on the performance. That is, the transmission power of all links is $p$. We let every transmitter (interferer) be active with probability $\rho_a$. Further, similar to~\cite{di2014stochastic}, we assume that transmitter of every link is spatially aligned with its intended receiver, so there is no beam-searching phase, see~\cite{Shokri2015Beam} for more information on how to model beam-searching overhead and to evaluate its impact on the network throughput.
If there is no obstacle on the link between transmitter $i$ and the origin, we say that transmitter $i$ has LoS condition respect to the typical receiver, otherwise it is in non-LoS (NLoS) condition. We consider a distance-dependent path-loss with exponent $\alpha$, as commonly assumed for MAC level performance evaluations~\cite{singh2009blockage,Singh2011Interference}. This simple model allows deriving tractable closed-form expressions for the collision probability and for the throughput, and at the same time, enables us to draw general conclusions about the network operating regime.

We use the \emph{protocol model} of interference~\cite{cardieri2010modeling}, which is common in the MAC layer analysis~\cite{Singh2011Interference,cardieri2010modeling}. In this model, for a given distance between a reference receiver and its transmitter, a \emph{collision}\footnote{Note that ``collision'' is defined as the outage event due to strong interference from other transmitters. Note that an outage can also occur due to low signal-to-noise ratio (SNR) even without any interference.} occurs if there is at least another interfering transmitter no farther than a certain distance of the reference receiver, hereafter called \emph{interference range}.
Besides its simplicity, recent study~\cite{Shokri2015Right} reveals that the special characteristics of mmWave networks makes the protocol model quite accurate in those networks. Essentially, as the probability of having LoS condition on a link decreases exponentially with the distance~\cite{Rappaport2015wideband}, far away transmitters will be most probably blocked (in NLoS condition) and therefore cannot contribute in the interference a receiver observes. It follows that we may consider only the impact of spatially close interferers with negligible loss in the accuracy of the interference model.

At the MAC layer, the beamforming is represented by using an ideal sector antenna pattern~\cite{di2014stochastic}, where the antenna gain is a constant for all angles in the main lobe and equal to a smaller constant in the side lobe.
This model allows capturing the interplay between antenna gain and half power beamwidth. Moreover, recent studies show that the interferers that are aligned with the typical receiver are dominant interferers and can cause collision~\cite{di2014stochastic}. Assuming the same operating beamwidth $\theta$ for all devices in both transmission and reception modes, neglecting the sidelobe radiations, and considering 2D beamforming, antenna gain for each transmitter/receiver is $2 \pi /\theta$ on the main lobe~\cite{Shokri2015Beam}. Further, we recall that the penetration loss due to solid materials in mmWave frequencies is so high (e.g., 35~dB due to the human body~\cite{Rangan2014Millimeter}) that the typical receiver can receive interference only from transmitters with LoS condition, called \emph{LoS interferers}.


\noindent \textbf{Blockage model:} As the operating beamwidth becomes narrower, the events of observing obstacles on the link between the typical receiver and individual interferers become more and more correlated, so the LoS condition for different interferers becomes correlated. One obstacle can block many interferes located very close to each other in the angular domain. Therefore, the assumption of independent LoS conditions on the links among the typical receiver and different interferers, as considered in~\cite{di2014stochastic}, is not adequate for mmWave systems. This inaccuracy increases as the number of links increases or if transmitters appear in spatial clusters. Assume that the center of obstacles follow a homogeneous Poisson point process with density $\lambda_o$ independent of the communication network. For sake of simplicity, we may use obstacle to refer the center of that obstacle. To capture the correlation among LoS conditions, we use the following blockage model: we define a \emph{coherence angle} $\theta_c$ over which the LoS conditions are correlated. That is, inside a coherence angle, an obstacle blocks all the interferers behind itself, so there is no LoS conditions in distances $d \geq l$ respect to the origin and consequently no LoS interferers, if there is an obstacle at distance $l$. However, there is no correlation between LoS condition events in different coherence angle intervals, simply in different circle sectors with angle $\theta_c$. The coherence angle depends on the average size and density of obstacles in the environment.

\section{Performance Evaluation}\label{sec: performance-evaluation}

\subsection{Collision Analysis}\label{sec: collision-analysis}
We denote by $p$ the transmission power, by $a$ the average channel attenuation at reference distance 1~meter, by $d_{\max}$ the interference range, by $\beta$ the minimum SINR threshold at the typical receiver, and by $\sigma$ the noise power. The interference range $d_{\max}$ is defined as the maximum distance an interferer can be from the receiver and still cause collision/outage.
Remarkably, although this model may be oversimplified in conventional networks (below 6~GHz), we have shown in~\cite{Shokri2015Right} that it is extremely accurate for mmWave networks.
The SINR at the typical receiver due to transmissions of the intended transmitter at distance $L$ and an interferer located at distance $d$, with LoS condition and aligned transmitter/receiver, is
\begin{equation*}
  \frac{p \left( 2 \pi/ \theta \right)^2 a L^{-\alpha}}{p \left( 2 \pi/\theta \right)^2 a d^{-\alpha} +\sigma} \:.
\end{equation*}
Comparing the above SINR expression to $\beta$, we get that the interference range
\begin{equation}\label{eq: d-max}
d_{\max} = \left( \frac{L^{-\alpha}}{\beta} - \frac{\sigma}{ap}\left( \frac{\theta}{2\pi} \right)^{2}\right)^{-1/\alpha} \:.
\end{equation}

A transmitter at distance $d$ from the typical receiver can cause collision provided that the following conditions hold: (a) it is active, (b) the typical receiver is inside its main lobe, (c) it is inside the main lobe of the typical receiver, (d) it is located inside the interference range $d \leq d_{\max}$, and (e) it is in LoS condition with respect to the typical receiver. These conditions are illustrated in Fig.~\ref{fig: IntRegion}, where the typical transmitter, interferers, and obstacles are represented by a green circle, red triangles, and blue rectangles, respectively.
Also, the highlighted part is the sector from which the typical receiver is receiving signal.
Interferers~1, 2, and 3 cannot cause collision at the typical receiver due to condition~(c),~(d), and~(e), respectively.

Due to random deployment of the devices, the probability that the typical receiver is inside the main lobe of an active transmitter is $\theta/2 \pi$. Therefore, if the density of transmitters is $\lambda_t$ and the average probability of being active for every transmitter is $\rho_a$, the interferers for which conditions~(a) and~(b) hold follow a Poisson point process with density $\lambda_I = \rho_a \lambda_t \theta/ 2 \pi$ per unit area. Conditions~(c) and~(d) reduce the area over which a potential interferer can cause collision. For condition~(e), we need to elaborate the blockage model. The typical receiver observes $k = \lceil \theta / \theta_c \rceil $ sectors, each with angle $\theta_c$, where $\lceil \cdot \rceil$ is the ceiling function. For the sake of simplicity, we assume that $ \theta / \theta_c $ is an integer; however, the analysis can be extended, with more involved calculations, to the general case.
We take the general assumption that the typical transmitter is uniformly distributed in the circle sector with angle $\theta$ that the typical receiver is pointing. Therefore, having a fix coordinate for the typical transmitter is a special case of our analysis. It is straightforward to see that the typical transmitter is located in one of these $k$ sectors with uniform distribution and its radial distance to the typical receiver $L$ is a continuous random variable with probability density function $f_{L}(\ell) = 2\ell / d_{\max}^{2}$. Without loss of generality, we assume that the typical transmitter is in sector $k$. It means that we have a combination of interferers and obstacles in the first $k-1$ sectors. In the last sector, we cannot have any obstacle in the circle sector with angle $\theta_c$ and radius $\ell$, as the typical transmitter in $\ell$ should be in the LoS condition, otherwise the typical link will not be established and collision cannot happen. Further dividing the last sector into two sub-sectors, corresponding to the distances $\left( 0, \ell \right]$ and $\left( \ell , d_{\max}\right]$, the first sub-sector contains only interferers, whereas the second one has both interferers and obstacles.
In the following, we first derive the probability of receiving collision from individual sectors and then compute the collision probability in general.

\begin{figure}[!t]
\centering
  \includegraphics[width=0.84\columnwidth]{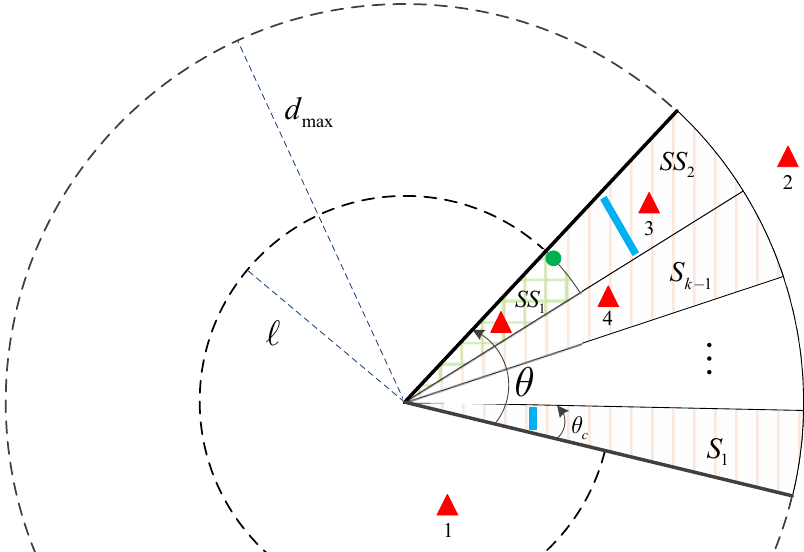}\\

  \caption{Hatched lines show potential interference zone. Operating beamwidth $\theta$ is divided into $k$ sectors of angle $\theta_c$. The typical receiver is on the origin. The typical transmitter, shown by a green circle, is on sector $k$ at distance $\ell$ of the typical receiver. $S_i$ shows sector $1 \leq i \leq k-1$. ${SS}_1$ and ${SS}_2$ are two sub-sectors of sector $k$. Zones with orange hatched lines have both random interferers and obstacles, represented by a red triangle and a blue rectangle. Zones with green hatched lines have only random interferers. $d_{\max}$ is the interference range.}
  \label{fig: IntRegion}
\end{figure}

Let $A_d$ be the area of a circle sector with radius $d$ and angle $\theta_c$.
Let $I_{(1)}$ and $O_{(1)}$ be the distance of the nearest interferer and obstacle to the origin, respectively. Given sector $s$, $1 \leq s \leq k-1$, PDFs of $I_{(1)}$ and $O_{(1)}$ are $f_{I_{(1)}}(x) = \lambda_I \theta_c x e^{-\lambda_I x^2 \theta_c/2 }$ and $f_{O_{(1)}}(y) = \lambda_o \theta_c y e^{-\lambda_o y^2 \theta_c/2 }$, respectively~\cite[Equation~(8)]{Haenggi2005onDistances}, where $\lambda_I$ and $\lambda_o$ are the density of potential interferers and obstacles per unit area. Under the protocol model, the typical receiver experiences LoS interference, with probability $\Pr[\mathrm{LI}]$, provided that $I_{(1)} < O_{(1)}$ while $0 \leq I_{(1)} \leq d_{\max}$. Considering independence of the interferer and obstacles processes, we have
\begin{align}\label{eq: LoSISec1:K-1}
  \Pr[\mathrm{LI}] & =
  \int_{x=0}^{d_{\max}} \int_{y=x}^{\infty} \! \lambda_I \theta_c x e^{-\lambda_I x^2 \theta_c/2 } \lambda_o \theta_c y e^{-\lambda_o y^2 \theta_c/2 } \, \mathrm{d}x \mathrm{d}y    \nonumber \\
   & = \lambda_I  \theta_{c} \int_{x=0}^{d_{\max}} x e^{-\left( \lambda_I + \lambda_o \right) x^2 \theta_c/2 } \, \mathrm{d}x \nonumber  \\
   & = \frac{\lambda_I}{\lambda_o + \lambda_I} \left( 1 - e^{-\left( \lambda_o + \lambda_I \right) A_{d_{\max}}} \right) \:,
\end{align}
where $\lambda_I = \rho_a \lambda_t \theta / 2 \pi$ and $A_{d_{\max}} = \theta_c d_{\max}^{2}/2$.

To find the probability of receiving LoS interference from sector $k$, we first note that sector $k$ consists of two sub-sectors, corresponding to the distances $\left( 0, \ell \right]$ and $\left( \ell , d_{\max}\right]$. In the first sub-sector, there is no obstacle, whereas we have regular appearance of the obstacles in the second sub-sector, see Fig.~\ref{fig: IntRegion}. Noting that $O_{(1)} \geq \ell$ in the last sector, the probability of LoS interference from sector $k$ is
\begin{align}\label{eq: LoSIsubsec1}
& \int_{x=0}^{\ell^{-}} \int_{y=\ell}^{\infty} \! \lambda_I \theta_c x e^{-\lambda_I x^2 \theta_c/2 } \lambda_o \theta_c y e^{-\lambda_o y^2 \theta_c/2 } e^{\lambda_o A_{\ell}} \, \mathrm{d}x \mathrm{d}y \nonumber \\
& \hspace{3.5mm} + \int_{x=\ell}^{d_{\max}} \int_{y=x}^{\infty} \! \lambda_I \theta_c x e^{-\lambda_I x^2 \theta_c/2 } \lambda_o \theta_c y e^{-\lambda_o y^2 \theta_c/2 } e^{\lambda_o A_{\ell}}\, \mathrm{d}x \mathrm{d}y    \nonumber \\
& = 1 - e^{- \lambda_I A_{\ell}} + \frac{\lambda_Ie^{\lambda_o A_{\ell}}}{\lambda_o + \lambda_I} \left( e^{-\left( \lambda_o + \lambda_I \right) A_{\ell}} - e^{-\left( \lambda_o + \lambda_I \right) A_{d_{\max}}} \right) \:.
\end{align}

\begin{figure*}[!t]
\normalsize
\begin{equation}\label{eq: CollProbFinal}
\rho_c = 1 - \hspace{-0.5mm} \int_{\ell=0}^{d_{\max}} \! \left( \frac{ \lambda_o + \lambda_I e^{- \left( \lambda_o + \lambda_I \right) \theta_c d_{\max}^{2}/2}}{\lambda_o + \lambda_I} \right)^{\hspace{-1mm} \lceil \theta / \theta_c \rceil -1} \hspace{-1mm} \left( e^{- \lambda_I \theta_c \ell^{2}/2} - \frac{\lambda_I  e^{\lambda_o A_{\ell}}}{\lambda_o + \lambda_I} \hspace{-0.7mm} \left( e^{-\left( \lambda_o + \lambda_I \right) \theta_c \ell^{2}/2} - e^{-\left( \lambda_o + \lambda_I \right) \theta_c d_{\max}^{2}/2} \right)\right) \frac{2\ell}{d_{\max}^{2}}  \mathrm{d}\ell .
\end{equation}
\hrulefill
\vspace*{-0.1pt}
\begin{equation}\label{eq: CollProbCond}
\rho_{c \mid \ell} = 1 - \left( \frac{ \lambda_o + \lambda_I e^{- \left( \lambda_o + \lambda_I \right) A_{d_{\max}}}}{\lambda_o + \lambda_I} \right)^{\lceil \theta / \theta_c \rceil -1} \left( e^{- \lambda_I A_{\ell}} - \frac{\lambda_I  e^{\lambda_o A_{\ell}}}{\lambda_o + \lambda_I} \left( e^{-\left( \lambda_o + \lambda_I \right) A_{\ell}} - e^{-\left( \lambda_o + \lambda_I \right) A_{d_{\max}}} \right) \right) \:.
\end{equation}
\vspace*{-0.1pt}
\hrulefill
\begin{equation}\label{eq: CollProbBounds}
1 - \left( \frac{ \lambda_o + \lambda_I e^{- \left( \lambda_o + \lambda_I \right) \theta_c d_{\max}^{2}/2}}{\lambda_o + \lambda_I} \right)^{\lceil \theta / \theta_c \rceil } \leq  \rho_c \leq 1 - e^{- \lambda_I \theta_c d_{\max}^{2}/2} \left( \frac{ \lambda_o + \lambda_I e^{- \left( \lambda_o + \lambda_I \right) \theta_c d_{\max}^{2}/2}}{\lambda_o + \lambda_I} \right)^{\lceil \theta / \theta_c \rceil -1} \:.
\end{equation}
\vspace*{-1mm}
\hrulefill
\end{figure*}

\begin{prop}\label{prop: Collision Probability_Bounds}
Let $\lambda_t$ and $\lambda_o$ denote the density of the transmitters and obstacles per unit area. Let $\rho_a$ be the probability that a transmitter is active. Consider blockage and interference models, described in Fig.~\ref{fig: IntRegion}. Let $d_{\max}$, $\theta$, and $\theta_c$ be the interference range, operating beamwidth, and coherence angle, respectively. Then, the collision probability, denoted by $\rho_c$, is given by~\eqref{eq: CollProbFinal} on the top of page~\pageref{eq: CollProbFinal}, where $\lambda_I = \rho_a \lambda_t \theta / 2 \pi$.
\end{prop}

\begin{proof}
Given that the typical link at length $\ell$ is established, the conditional collision probability, denoted by $\rho_{c \mid \ell}$, is equal to the probability of having at least one LoS interferer, irrespective of the sectors (sub-sectors) in which the LoS interferer(s) are. To derive the collision probability, we first find its complementary, i.e., the probability of having no LoS interferer in any sector. The latter is equal to complementary of the event of having collision in any sector, given by~\eqref{eq: LoSISec1:K-1} and \eqref{eq: LoSIsubsec1}. Considering mutual independence of different sectors, $\rho_{c \mid \ell}$ is given by~\eqref{eq: CollProbCond}.
The last step of characterizing the collision probability is taking an average of $\rho_{c \mid \ell}$ over distribution of $\ell$, which is distributed with PDF $f_L(\ell) = 2\ell / d_{\max}^{2}$ within $(0,d_{\max}]$. This completes the proof.
\end{proof}

Observe that the collision probability given by~\eqref{eq: CollProbFinal} is strictly increasing with $\ell$. Therefore, we can derive its lower and upper bounds by substituting $\ell = 0$ and $\ell = d_{\max}$ into~\eqref{eq: CollProbFinal}, respectively. The bounds are given by~\eqref{eq: CollProbBounds}. These bounds provide tractable approximations of the collision probability, given by~\eqref{eq: CollProbFinal}, to design pessimistic/optimistic collision-aware resource allocation strategies.

\subsection{Noise-Limited or Interference-Limited}\label{sec: noise-int-limited}
Using the closed-form expression of the collision probability, established in~\eqref{eq: CollProbFinal}, we now derive average per-link throughput and analyze the regime at which the network operates. We first note that the typical transmitter is active with probability $\rho_a$. Its transmission to the receiver at distance $\ell$ is successful if there is no blockage, which occurs with probability $ e^{- \lambda_o A_{\ell}} $, and no collision, which occurs with probability $\left( 1 - \rho_{c \mid \ell} \right)$. Assuming transmission of one packet per slot, the average throughput will be equal to the average success probability. Hence, the average per-link throughput for a slotted ALOHA system, denoted by $r_{_{\text{S-ALOHA}}}$, is
\begin{equation}\label{eq: MACthroughputFinal0}
r_{_{\text{S-ALOHA}}} = \int_{\ell=0}^{d_{\max}} \! \rho_a e^{- \lambda_o A_{\ell}} \left( 1 - \rho_{c \mid \ell} \right) \frac{2\ell}{d_{\max}^{2}}\, \mathrm{d}\ell \quad  {\mbox{packets/slot}}  \:,
\end{equation}
which can be tightly bounded using~\eqref{eq: CollProbBounds}. For a given $\rho_a$, the throughput is uniquely determined by the collision probability. It follows that we can study the collision probability, instead of the throughput, to identify the network operating regime.

By definition, we are in the \emph{noise-limited} regime if the collision probability is too small for given density of the obstacles, density of the transmitters, and operating beamwidth, among the main parameters. However, if there is at least a LoS interferer, which limits the throughput performance of the network, we are in the \emph{LoS interference-limited} regime. This suggests the following conclusion: \emph{A mmWave network with directional communication may have transitional region, i.e., the region in which there is a transition from a noise-limited regime to a LoS interference-limited one}. This region depends on the density of interferers and obstacles, transmission probability, operating beamwidth, transmission powers, and coherence angle. Due to lack of space, we only present the following proposition without proof, and will use it in the following sections. Detailed proofs are available in the extended version of this paper in~\cite{shokri2015collision}.

\begin{prop}\label{prop: TDMAthroughputPerformance}
Let $\lambda_t$ and $\lambda_o$ denote the density of the transmitters and obstacles per unit area. Let $\rho_a$ be the probability that a transmitter is active. Consider blockage and interference models, described in Fig.~\ref{fig: IntRegion}. Let $d_{\max}$, $\theta$, and $\theta_c$ be the interference range, operating beamwidth, and coherence angle, respectively.
Let $A$ denote the area over which scheduler (either slotted ALOHA or TDMA) regulates the transmissions of the transmitters. Let $r_{_{\text{TDMA}}}$, ${\mathrm{ASE}}_{\text{TDMA}}$, ${\mathrm{ASE}}_{\text{S-ALOHA}}$ be per-link throughput and area spectral efficiency (ASE) of TDMA and slotted ALOHA. Then, we have
\begin{equation}\label{eq: TDMALinkThroughput}
r_{_{\text{TDMA}}} =  \left( \frac{1 - e^{-\lambda_t A}}{\lambda_t A} \right) \left( \frac{1-e^{-\lambda_o A_{d_{\max}}}}{\lambda_o A_{d_{\max}}} \right) \:,
\end{equation}
\begin{equation}\label{eq: TDMA-ASE}
{\mathrm{ASE}}_{\text{TDMA}} =  \frac{1-e^{-\lambda_o A_{d_{\max}}}}{A \lambda_o A_{d_{\max}}} \:,
\end{equation}
and
\begin{equation}\label{eq: AreaSpecEffic}
{\mathrm{ASE}}_{\text{S-ALOHA}} = \frac{1 + A \lambda_t}{A} \int_{\ell=0}^{d_{\max}} \! \frac{2\rho_a}{d_{\max}^{2}}  e^{- \lambda_o A_{\ell}} \left( 1 - \rho_{c \mid \ell} \right) \ell \, \mathrm{d}\ell \:,
\end{equation}
where $A_{\mathrm{x}} = \theta_c \mathrm{x}^2/2$, and $\rho_{c \mid \ell}$ is given by~\eqref{eq: CollProbCond}.
\end{prop}

\subsection{Contention-Based or Contention-Free}\label{sec: contention-based-free}
The transitional behavior of interference in mmWave networks, along with pencil-beam operation, demands rethinking the proper resource allocation and interference management strategies for mmWave networks.
The noise-limited regime implies negligible multiuser interference, so activating all links at the same time provides the maximum MAC throughput. This indeed means that the throughput performance of one of the simplest collision-based protocols (slotted ALOHA) is almost equivalent to that of the optimal collision-free resource allocation, i.e., spatial TDMA also called STDMA~\cite{nelson1985spatial,Shokri2015Beam,bjorklund2003resource}. Finding the optimal resource allocation based on STDMA is an NP-hard problem in general~\cite{bjorklund2003resource} and requires the exact network topology~\cite{Shokri2015Beam}.
Discovering the topology (even partially), in turn, requires exchanging several control messages, which may be overwhelming in mmWave networks due to the alignment overhead~\cite{Shokri2015Beam,shokri2015mmWavecellular}.
To relax the computational and signaling complexities of STDMA, existing mmWave standards use TDMA as the main resource allocation strategy~\cite{802_15_3c,802_11ad}. However, TDMA does not support concurrent transmissions leading to a substantially lower network throughput --as low as 0.2 of the maximum one for a mmWave WPAN with only 10 links~\cite{Shokri2015Beam}-- compared to slotted ALOHA.
As the multiuser interference level increases (moving to the interference-limited regime), e.g., due to wider operating beamwidths or higher link densities, efficient contention-free channel access strategy will advantageously maximize the throughput performance, while guaranteeing collision-free communications, which is necessary for specific applications. This guarantee is not available in the contention-based channel access strategies.

The transitional behavior of interference in mmWave networks indicates inefficacy of existing standards and suggests a dynamic incorporation of both contention-based and contention-free phases in the resource allocation. The current mmWave standards such as IEEE~802.15.3c and IEEE~802.11ad adopt the existing resource allocation approaches that were originally developed for interference-limited networks. In particular, they introduce a contention-based phase mainly to register channel access requests of the devices inside the mmWave network. These requests are served on the following TDMA-based contention-free phase. In fact, though some data messages with low QoS requirements may be transmitted in the contention-based phase, the network traffic is mostly served in the contention-free phase irrespective of the actual network operating regime. Instead, we can leverage the transitional behavior of mmWave networks and dynamically serve the network traffic partially on the contention-based and partially on the contention-free phase, according to the actual network operating regime.
In a noise-limited regime, we can deliver most of the traffic on the contention-based phase due to negligible multiuser interference. In an interference-limited regime, many links may register their channel access requests to be scheduled on the contention-free phase. Using flexible phase duration, adjusted according to the collision level of the networks, we can reduce the use of inefficient contention-free phase, improve the network throughput (especially as the network goes to the noise-limited regime), and also guarantee the collision-free communication. Developing a proper adaptive hybrid MAC protocol for mmWave networks will be undertaken in our future studies.


\section{Numerical Results}\label{sec: numerical-results}
We simulate an ad~hoc network at 60~GHz with random number of links and obstacles with densities $\lambda_t$ and $\lambda_o$, respectively. We let all the links to be active with transmission power of 2.5 mW and assume a coherence angle of $\theta_c = 5 \degree$. Using Monte Carlo simulations, we evaluate the performance metrics by averaging over $10^6$ random topologies. Moreover, to evaluate the validity of both the blockage model and the throughput analysis, we build a system-level mmWave emulator in ns3 environment. In this emulator, we consider a random number of aligned mmWave links (aligned transmitter-receiver pairs), all operating with the same beamwidth $\theta$. The obstacles are in the shape of lines with random orientations and their lengths are uniformly distributed between 0 and 1~m. Every transmitter generates traffic with constant bit rate (CBR) 384~Mbps, the size of all packets is 5~kB, time slot duration is 100~$\mu$s, transmission rate is 1 packet per slot (link capacity around 1.5~Gbps), the transmitters have infinite buffer to save and transmit the packets, and the emulation time is 1 second.
Due to lack of space, we only show the impact of the most important parameters on the performance metrics, derived in~\eqref{eq: CollProbFinal}--\eqref{eq: AreaSpecEffic}. Detailed numerical results can be found in the extended version of this work in~\cite{shokri2015collision}.

\begin{figure}[!t]
	\centering
	\includegraphics[width=\columnwidth]{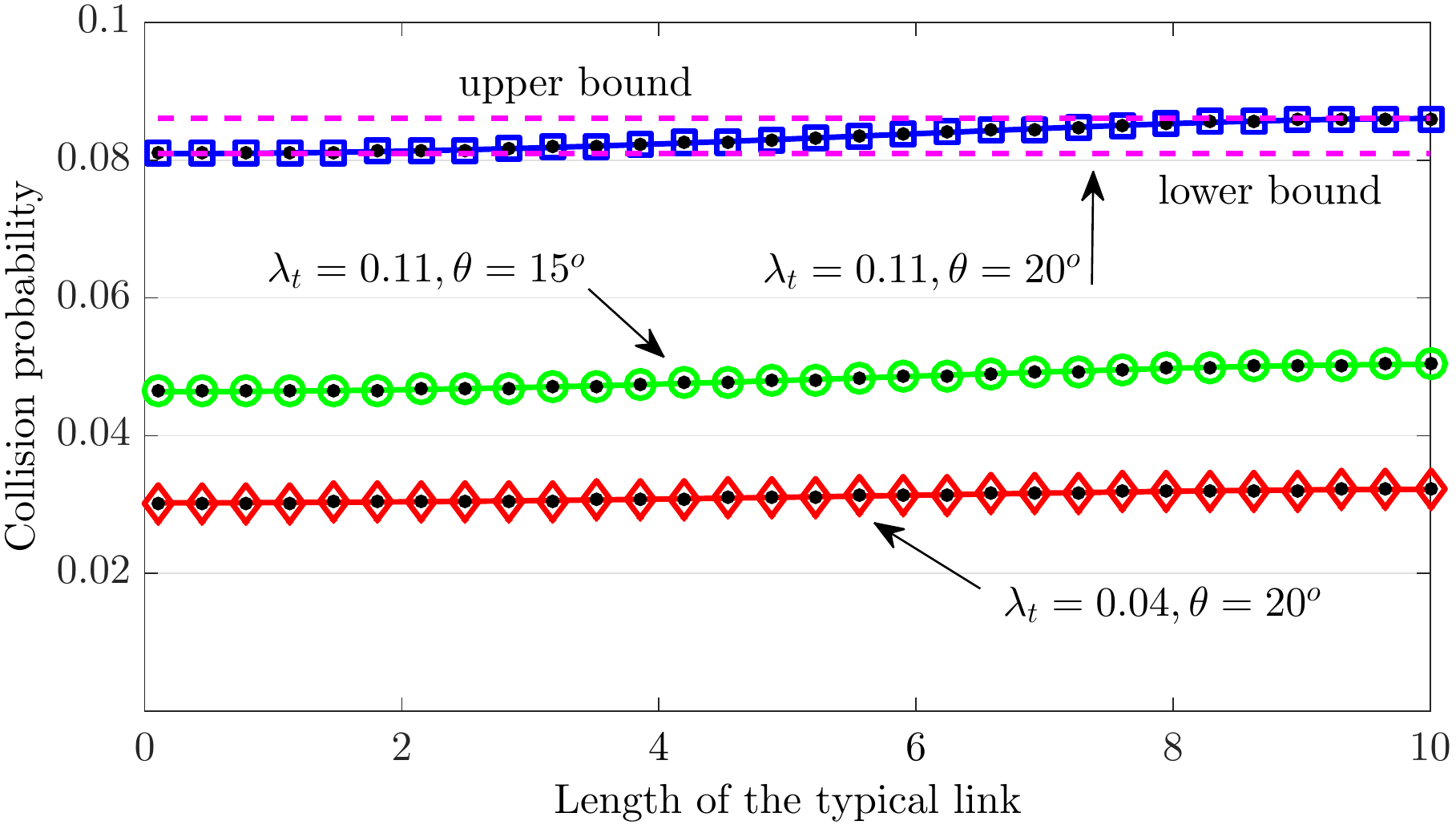}

    \vspace{-1mm}
    \caption{Collision probability as a function of the length of the typical link $\ell$, as computed by Equation~\eqref{eq: CollProbCond} and Monte Carlo simulations, marked by black circles. Upper and lower bounds are computed by Equation~\eqref{eq: CollProbBounds}.}
    \label{fig: CollisionProbVsL}
\end{figure}
Fig.~\ref{fig: CollisionProbVsL} shows collision probability against the length of the typical link $\ell$. As mentioned in Section~\ref{sec: collision-analysis}, the collision probability is an increasing function of $\ell$ with lower and upper bounds, formulated in~\eqref{eq: CollProbBounds}. First of all, analytical results well match numerical ones. Furthermore, both upper and lower bounds are tight for all examples considered, implying that the approximated closed-form bounds in~\eqref{eq: CollProbBounds} can be effectively used for pessimistic/optimistic MAC layer designs. For the example of 1 transmitter in a 3x3~${\text{m}}^2$ area and operating beamwidth $\theta = 20 \degree$, the maximum error due to those approximations, i.e., the difference between upper and lower bounds is only 0.005. This error reduces as the operating beamwidth or transmitter density reduce, as can be observed in Fig.~\ref{fig: CollisionProbVsL}.

\begin{figure}[!t]
	\centering
	\includegraphics[width=\columnwidth]{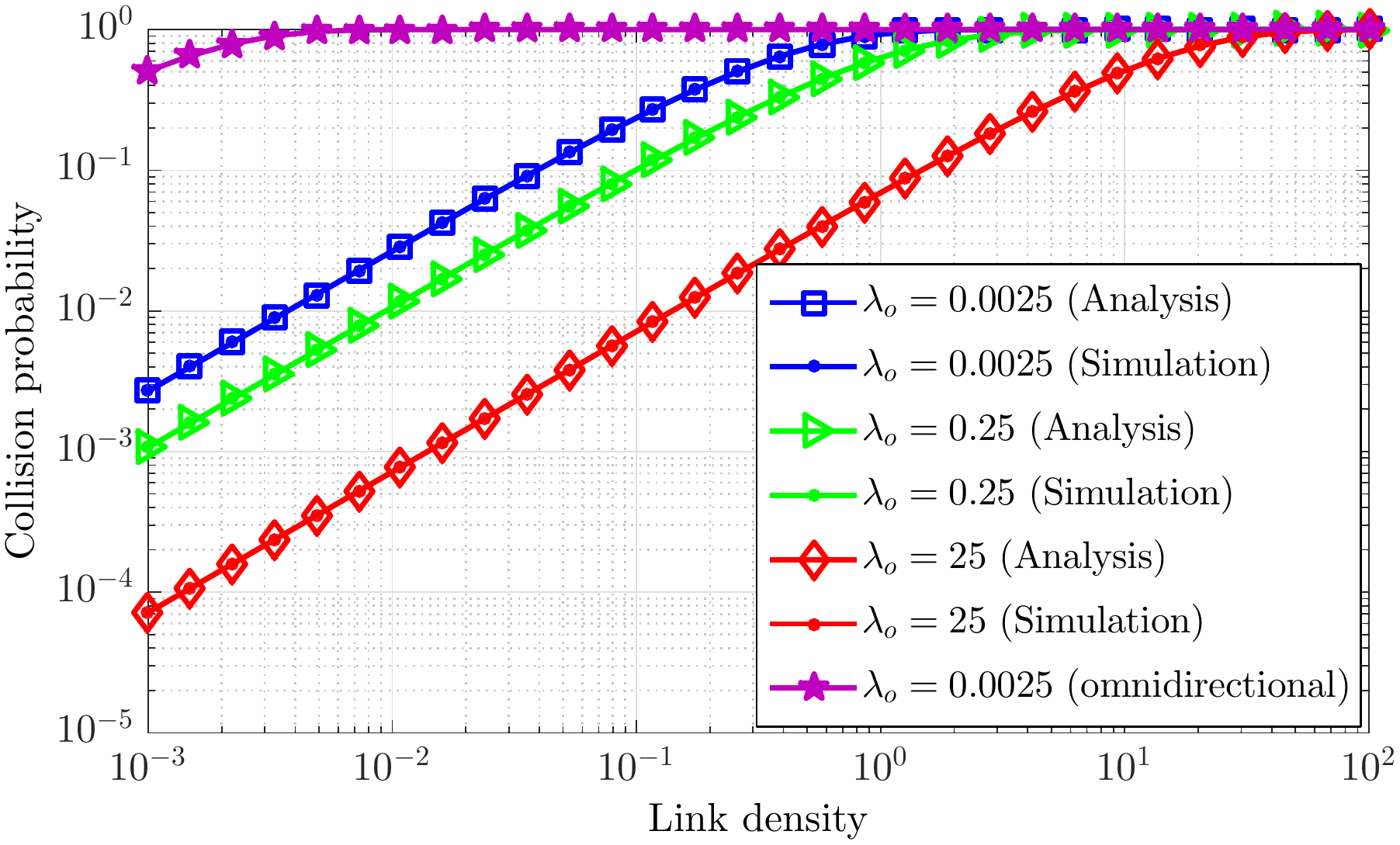}

	\vspace{-1mm}
\caption{Collision probability as a function of link density, as computed by Equation~\eqref{eq: CollProbFinal} and Monte Carlo simulations, marked by filled circles. Operating beamwidth is $\theta = 20 \degree$.}
	\label{fig: CollisionProb}
\end{figure}
Fig.~\ref{fig: CollisionProb} shows the collision probability as a function of transmitter/link density $\lambda_t$. Not surprisingly, increasing the link density increases the collision probability. Also, higher obstacle density increases blockage probability, so reduces the collision probability.
It is clear from this figure that the collision probability is non-negligible even for modest-size mmWave networks. For instance, for 1 transmitter in a 3x3~${\text{m}}^2$ area and 1 obstacle in a 20x20~${\text{m}}^2$ area, the collision probability is as much as 0.26. Increasing the density of obstacles to 1 obstacle in a 3x3~${\text{m}}^2$ area, which is not shown in Fig.~\ref{fig: CollisionProb} for the sake of clarity, the collision probability reduces to 0.17, which is still high enough to invalidate the assumption of having a noise-limited mmWave network. Moreover, there is a transition from the noise-limited regime to the LoS interference-limited one in all curves.
For benchmarking purposes, we also simulate a network with omnidirectional communications. Fixing all other parameters, we increase the transmission power to achieve the same interference range as the corresponding directional communication and investigate the collision probability. As shown in Fig.~\ref{fig: CollisionProb}, traditional networks with omnidirectional communications always experience an interference-limited regime without any transitional behavior.

\begin{figure}[!t]
	\centering
    \subfigure[]{
	\includegraphics[width=0.985\columnwidth]{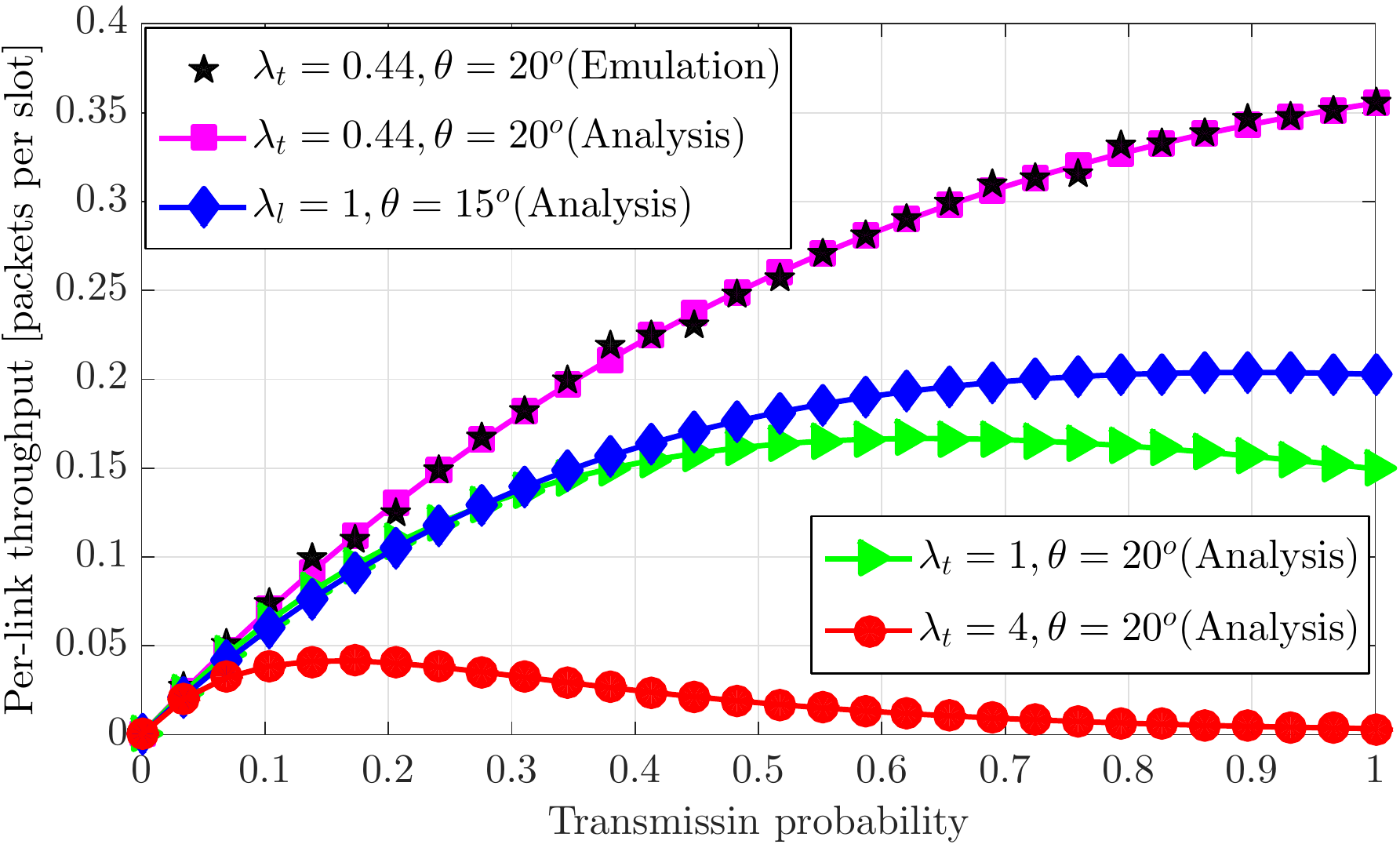}
	\label{fig: EffectiveMACThroughput}
	}
\subfigure[]{
\includegraphics[width=0.985\columnwidth]{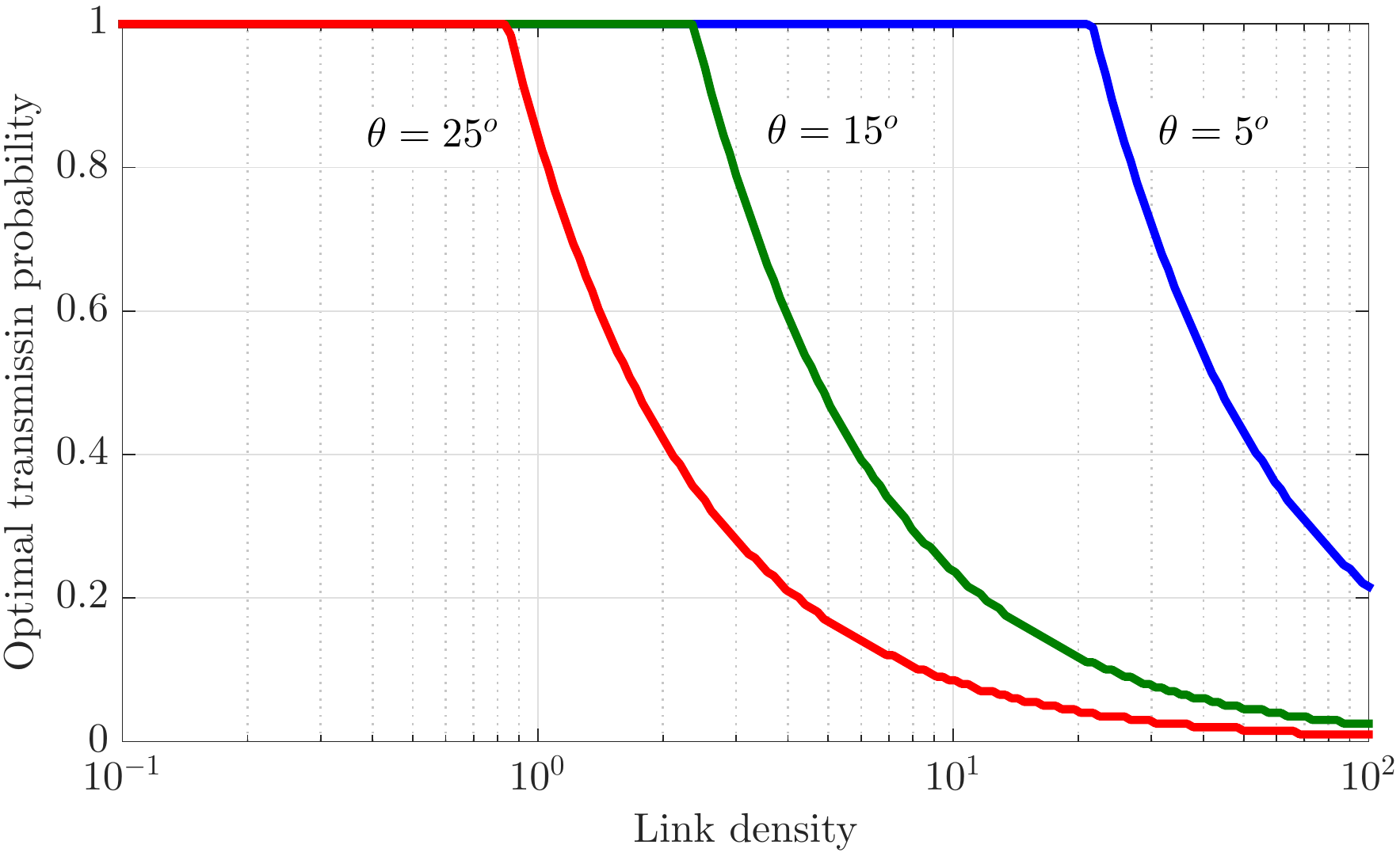}
	\label{subfig: OptimMACThr_OptimlTxProb}
	}
	
\vspace{-1mm}
\caption{Throughput analysis of slotted ALOHA:~\subref{fig: EffectiveMACThroughput} per-link throughput and \subref{subfig: OptimMACThr_OptimlTxProb} the optimal transmission probability.}
\label{fig: throughput-analysis}
\end{figure}
Per-link throughput of a mmWave ad~hoc network is illustrated in Fig.~\ref{fig: throughput-analysis}. First of all, there is a well match between the results obtained from the emulator and those from Equation~\eqref{eq: MACthroughputFinal0}, which confirms the validity of both blockage model and the throughput analysis. From Fig.~\ref{fig: EffectiveMACThroughput}, for relatively not so dense networks, e.g., 1 transmitter in a 1.5x1.5~${\text{m}}^2$ area ($\lambda_t = 0.44$), increasing the transmission probability is always beneficial, as the multiuser interference level is small enough that activating more links will not substantially reduce the average throughput of a link but increases the number of time slots over which the link is active. As the link density increases, higher collision probability introduces a tradeoff on increasing the transmission probability and reducing the interference. In a very dense network, e.g., with $\lambda_t = 4$, we should adopt a very small transmission probability to maximize the MAC throughput. Fig.~\ref{subfig: OptimMACThr_OptimlTxProb} shows the behavior of such an optimal transmission probability as a function of link density and operating beamwidth.
From Fig.~\ref{subfig: OptimMACThr_OptimlTxProb}, in many cases, the optimal transmission probability is 1, implying that the optimal transmission policy is activating all links simultaneously.
In fact, since there is a negligible multiuser interference in those cases, the performance of one of the simplest collision-based protocols (slotted ALOHA) is almost equivalent to the optimal collision-free resource allocation (STDMA) with much less signaling and computational overheads. However, as the operating beamwidth or the link density increase, we should adopt a very small transmission probability to decrease the contention level in slotted ALOHA, e.g., $\rho_a = 0.17$ for $\lambda_t = 3$ and $\theta = 25 \degree$. Alternatively, we can think of more regulated resource allocation strategies. Monitoring the collision level, we can develop an intelligent strategy to dynamically switch between contention-based and contention-free phases, which is subject of our future study.

\begin{figure}[!t]
	\centering
	\subfigure[]{
	\includegraphics[width=0.99\columnwidth]{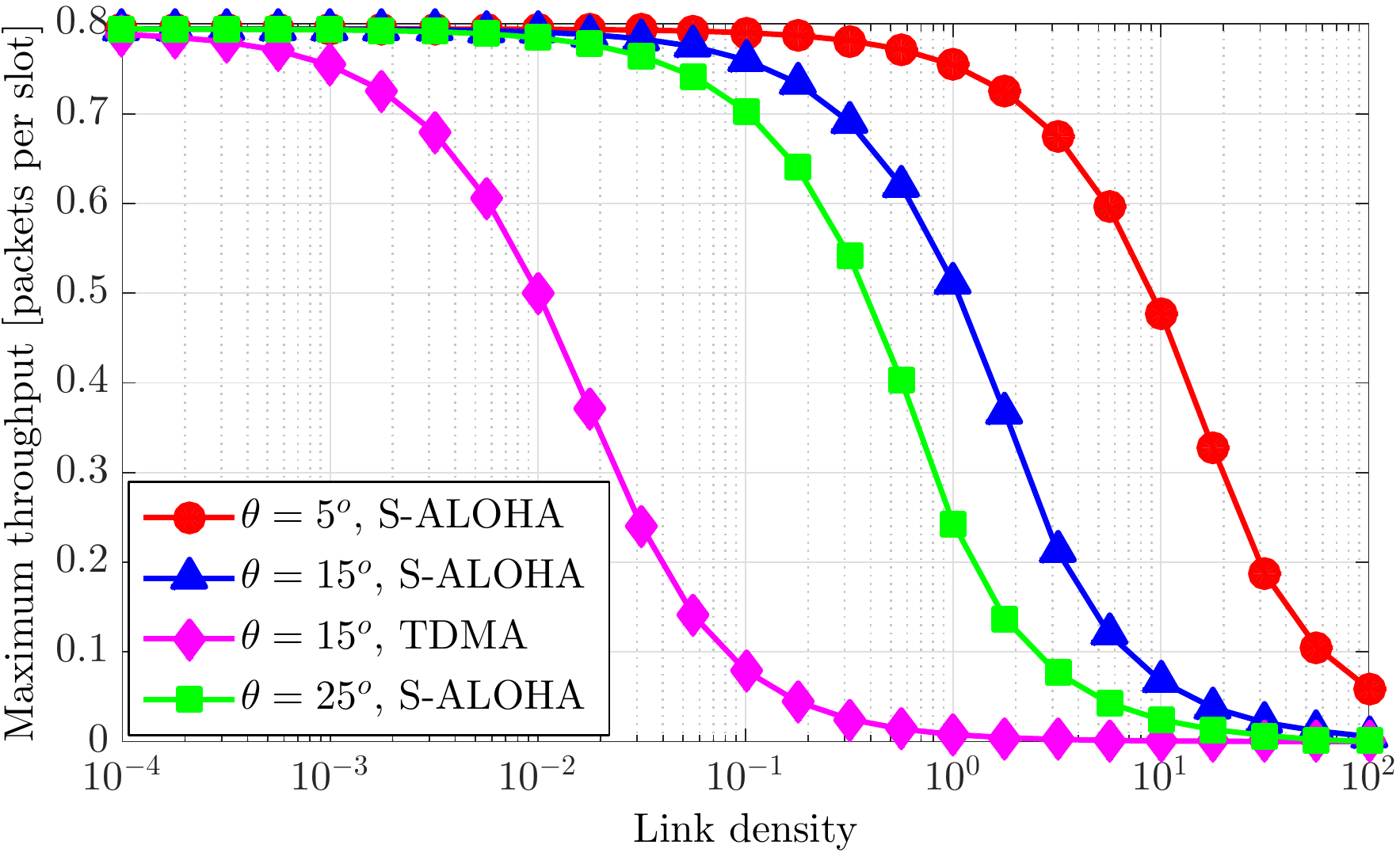}
		\label{subfig: OptimMACThr_MaxThr}
	}
    \subfigure[]{
	\includegraphics[width=\columnwidth]{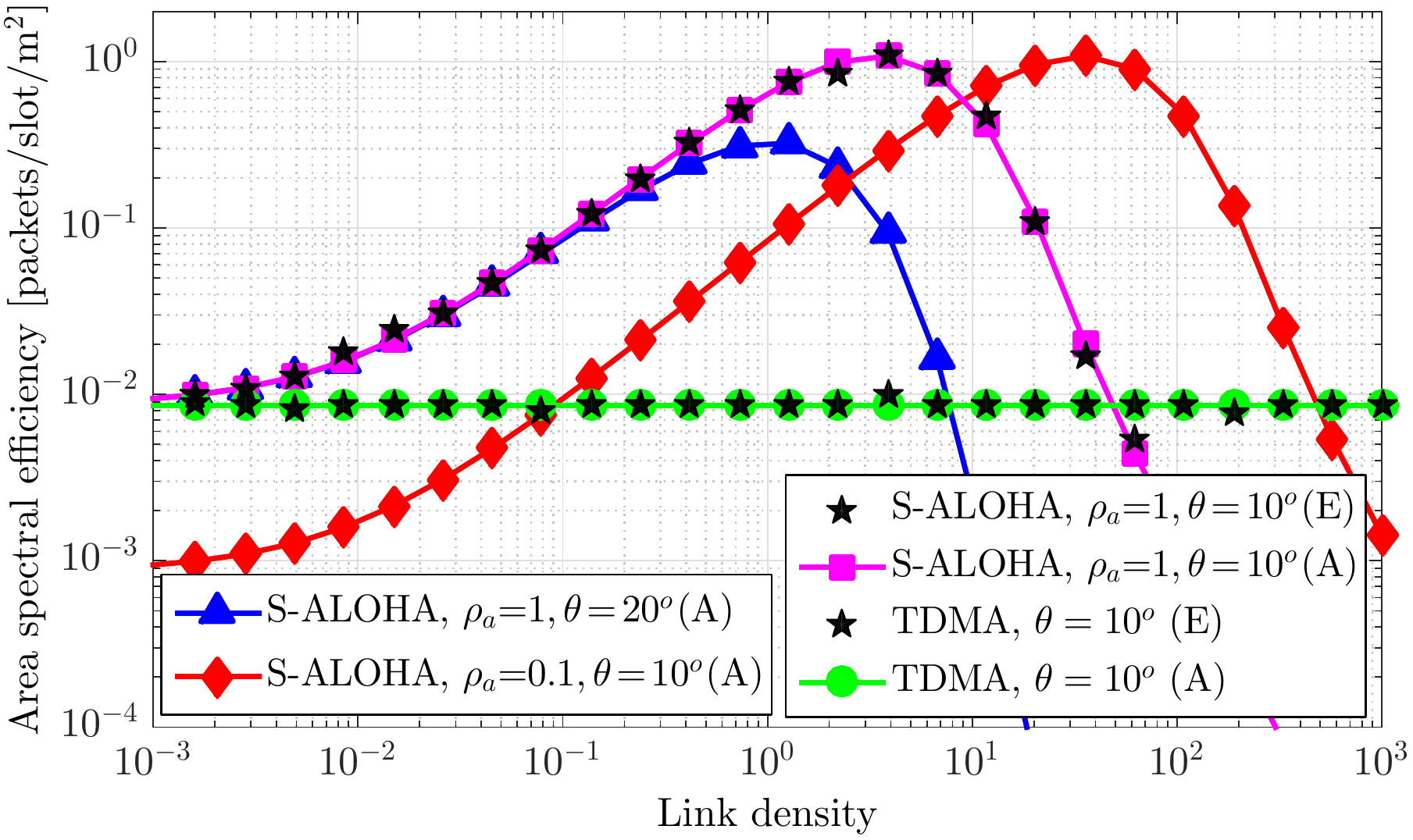}
		\label{subfig: ASE_LinkDensity}
	}
	
\vspace{-1mm}
    \caption{Performance comparison of slotted ALOHA and TDMA: ~\subref{subfig: OptimMACThr_MaxThr} per-link throughput and~\subref{subfig: ASE_LinkDensity} ASE. ``S-ALOHA'' stands for slotted ALOHA, ``(A)'' for (Analysis), and ``(E)'' for (Emulation).}
	\label{fig: OptimMACThr}
\end{figure}

Fig.~\ref{subfig: OptimMACThr_MaxThr} shows the maximum throughput of a link in slotted ALOHA, associated with the optimal transmission probability for $d_{\max} = 10$. First, per-link throughput in slotted ALOHA will be decreased with $\theta$, due to lower collision probability. Furthermore, slotted ALOHA significantly outperforms TDMA. The main reason is that TDMA realizes an orthogonal use of time resources, irrespective of the collision level, whereas slotted ALOHA re-uses all the time resources and benefits from spatial gain. This gain leads to 497\% and 2047\% throughput enhancements over TDMA for the cases of 1 transmitter in a 4x4~${\text{m}}^2$ and in a 2x2~${\text{m}}^2$ area with $\theta=15 \degree$, respectively. Note that, from Fig.~\ref{subfig: OptimMACThr_OptimlTxProb}, the optimal transmission probability is 1 in both cases, further highlighting simplicity of slotted ALOHA.
Both TDMA and slotted ALOHA show the same asymptotic zero throughput behavior, though with significantly different rates of convergence. Considering any arbitrary small $\zeta$ for the per-link throughput, from Fig.~\ref{subfig: OptimMACThr_MaxThr}, the per-link throughput of both TDMA and slotted ALOHA become lower than $\zeta$ for sufficiently large $\lambda_{t}$. However, slotted ALOHA reaches that point with almost two orders of magnitude more links in the network (e.g., see $\zeta = 0.1$), indicating its efficiency on handling massive wireless access in mmWave networks.

Fig.~\ref{subfig: ASE_LinkDensity} illustrates ASE of slotted ALOHA and TDMA as a function of link density. Again, there is a well coincidence among the analytical results obtained from Equations~\eqref{eq: TDMA-ASE} and~\eqref{eq: AreaSpecEffic} and those of the emulator.
Increasing the number of links of the network does not affect ASE of TDMA. The average network throughput of TDMA is slightly lower than one packet per slot, and it achieves the upper bound if the obstacle density goes to zero, see \cite[Corollary~2]{shokri2015collision}. Slotted ALOHA with transmission probability $\rho_a = 1$ provides the highest ASE, which is firstly increasing with the link density but then shows a strictly decreasing behavior once throughput loss due to the collision term overweighs the throughput enhancement due to the first term of~\eqref{eq: AreaSpecEffic}. For the example of $\rho_a = 1$ and $\theta = 10 \degree$, the optimal density of transmitters that maximizes ASE is, on average, 3.5 transmitters per square meter. This example number indeed implies that mmWave networks benefit from dense deployment. Slotted ALOHA with $\rho_a = 0.1$ outperforms that with $\rho_a=1$ in ultra dense WPANs ($\lambda_t>9$ in Fig.~\ref{subfig: ASE_LinkDensity}), as lower transmission probability leads to fewer active links in such networks.

\section{Conclusion}\label{sec: Conclusion}
Pencil-beam operation in millimeter wave (mmWave) networks reduces multiuser interference, which may lead to noise-limited mmWave networks.
In this paper, tractable closed-form expressions for collision probability, per-link throughput, and area spectral efficiency in a mmWave ad~hoc network with slotted ALOHA and TDMA schedulers were derived.
The analysis indicated that mmWave networks may not be necessarily noise-limited; rather, they show a transitional behavior from a noise-limited regime to an interference-limited one. This transitional behavior of interference necessitates new hybrid resource allocation procedures that consist of both contention-based and contention-free phases with adaptive phase duration. The duration of each phase depends on the actual network operating regime. The contention-based phase improves throughput performance, while the contention-free phase is still necessary to guarantee collision-free communications.

\bibliographystyle{IEEEtran}
\bibliography{References}

\end{document}

%% file: commands.tex

\newtheorem{defin}{Definition}
\newtheorem{theorem}{Theorem}
\newtheorem{prop}{Proposition}
\newtheorem{lemma}{Lemma}
\newtheorem{alg}{Algorithm}
\newtheorem{remark}{Remark}
\newtheorem{example}{Example}
\newtheorem{notations}{Notations}
\newtheorem{assumption}{Assumption}

\newcommand{\be}{\begin{equation}}
\newcommand{\ee}{\end{equation}}
\newcommand{\ba}{\begin{array}}
\newcommand{\ea}{\end{array}}
\newcommand{\bea}{\begin{eqnarray}}
\newcommand{\eea}{\end{eqnarray}}
\newcommand{\combin}[2]{\ensuremath{ \left( \ba{c} #1 \\ #2 \ea \right) }}
\newcommand{\diag}{{\mbox{diag}}}
\newcommand{\rank}{{\mbox{rank}}}
\newcommand{\dom}{{\mbox{dom{\color{white!100!black}.}}}}
\newcommand{\range}{{\mbox{range{\color{white!100!black}.}}}}
\newcommand{\image}{{\mbox{image{\color{white!100!black}.}}}}
\newcommand{\herm}{^{\mbox{\scriptsize H}}}  
\newcommand{\sherm}{^{\mbox{\tiny H}}}       
\newcommand{\tran}{^{\mbox{\scriptsize T}}}  
\newcommand{\tranIn}{^{\mbox{-\scriptsize T}}}  
\newcommand{\card}{{\mbox{\textbf{card}}}}
\newcommand{\asign}{{\mbox{$\colon\hspace{-2mm}=\hspace{1mm}$}}}
\newcommand{\ssum}[1]{\mathop{ \textstyle{\sum}}_{#1}}

\newcommand{\vbar}{\raisebox{.17ex}{\rule{.04em}{1.35ex}}}
\newcommand{\vbarind}{\raisebox{.01ex}{\rule{.04em}{1.1ex}}}
\newcommand{\D}{\ifmmode {\rm I}\hspace{-.2em}{\rm D} \else ${\rm I}\hspace{-.2em}{\rm D}$ \fi}
\newcommand{\T}{\ifmmode {\rm I}\hspace{-.2em}{\rm T} \else ${\rm I}\hspace{-.2em}{\rm T}$ \fi}
\newcommand{\B}{\ifmmode {\rm I}\hspace{-.2em}{\rm B} \else \mbox{${\rm I}\hspace{-.2em}{\rm B}$} \fi}
\newcommand{\Hil}{\ifmmode {\rm I}\hspace{-.2em}{\rm H} \else \mbox{${\rm I}\hspace{-.2em}{\rm H}$} \fi}
\newcommand{\C}{\ifmmode \hspace{.2em}\vbar\hspace{-.31em}{\rm C} \else \mbox{$\hspace{.2em}\vbar\hspace{-.31em}{\rm C}$} \fi}
\newcommand{\Cind}{\ifmmode \hspace{.2em}\vbarind\hspace{-.25em}{\rm C} \else \mbox{$\hspace{.2em}\vbarind\hspace{-.25em}{\rm C}$} \fi}
\newcommand{\Q}{\ifmmode \hspace{.2em}\vbar\hspace{-.31em}{\rm Q} \else \mbox{$\hspace{.2em}\vbar\hspace{-.31em}{\rm Q}$} \fi}
\newcommand{\Z}{\ifmmode {\rm Z}\hspace{-.28em}{\rm Z} \else ${\rm Z}\hspace{-.38em}{\rm Z}$ \fi}

\newcommand{\sgn}{\mbox {sgn}}
\newcommand{\var}{\mbox {var}}
\newcommand{\E}{\mbox {E}}
\newcommand{\cov}{\mbox {cov}}
\renewcommand{\Re}{\mbox {Re}}
\renewcommand{\Im}{\mbox {Im}}
\newcommand{\cum}{\mbox {cum}}

\renewcommand{\vec}[1]{{\bf{#1}}}     
\newcommand{\vecsc}[1]{\mbox {\boldmath \scriptsize $#1$}}     
\newcommand{\itvec}[1]{\mbox {\boldmath $#1$}}
\newcommand{\itvecsc}[1]{\mbox {\boldmath $\scriptstyle #1$}}
\newcommand{\gvec}[1]{\mbox{\boldmath $#1$}}

\newcommand{\balpha}{\mbox {\boldmath $\alpha$}}
\newcommand{\bbeta}{\mbox {\boldmath $\beta$}}
\newcommand{\bgamma}{\mbox {\boldmath $\gamma$}}
\newcommand{\bdelta}{\mbox {\boldmath $\delta$}}
\newcommand{\bepsilon}{\mbox {\boldmath $\epsilon$}}
\newcommand{\bvarepsilon}{\mbox {\boldmath $\varepsilon$}}
\newcommand{\bzeta}{\mbox {\boldmath $\zeta$}}
\newcommand{\boldeta}{\mbox {\boldmath $\eta$}}
\newcommand{\btheta}{\mbox {\boldmath $\theta$}}
\newcommand{\bvartheta}{\mbox {\boldmath $\vartheta$}}
\newcommand{\biota}{\mbox {\boldmath $\iota$}}
\newcommand{\blambda}{\mbox {\boldmath $\lambda$}}
\newcommand{\bmu}{\mbox {\boldmath $\mu$}}
\newcommand{\bnu}{\mbox {\boldmath $\nu$}}
\newcommand{\bxi}{\mbox {\boldmath $\xi$}}
\newcommand{\bpi}{\mbox {\boldmath $\pi$}}
\newcommand{\bvarpi}{\mbox {\boldmath $\varpi$}}
\newcommand{\brho}{\mbox {\boldmath $\rho$}}
\newcommand{\bvarrho}{\mbox {\boldmath $\varrho$}}
\newcommand{\bsigma}{\mbox {\boldmath $\sigma$}}
\newcommand{\bvarsigma}{\mbox {\boldmath $\varsigma$}}
\newcommand{\btau}{\mbox {\boldmath $\tau$}}
\newcommand{\bupsilon}{\mbox {\boldmath $\upsilon$}}
\newcommand{\bphi}{\mbox {\boldmath $\phi$}}
\newcommand{\bvarphi}{\mbox {\boldmath $\varphi$}}
\newcommand{\bchi}{\mbox {\boldmath $\chi$}}
\newcommand{\bpsi}{\mbox {\boldmath $\psi$}}
\newcommand{\bomega}{\mbox {\boldmath $\omega$}}

\newcommand{\bolda}{\mbox {\boldmath $a$}}
\newcommand{\bb}{\mbox {\boldmath $b$}}
\newcommand{\bc}{\mbox {\boldmath $c$}}
\newcommand{\bd}{\mbox {\boldmath $d$}}
\newcommand{\bolde}{\mbox {\boldmath $e$}}
\newcommand{\boldf}{\mbox {\boldmath $f$}}
\newcommand{\bg}{\mbox {\boldmath $g$}}
\newcommand{\bh}{\mbox {\boldmath $h$}}
\newcommand{\bp}{\mbox {\boldmath $p$}}
\newcommand{\bq}{\mbox {\boldmath $q$}}
\newcommand{\br}{\mbox {\boldmath $r$}}
\newcommand{\bs}{\mbox {\boldmath $s$}}
\newcommand{\bt}{\mbox {\boldmath $t$}}
\newcommand{\bu}{\mbox {\boldmath $u$}}
\newcommand{\bv}{\mbox {\boldmath $v$}}
\newcommand{\bw}{\mbox {\boldmath $w$}}
\newcommand{\bx}{\mbox {\boldmath $x$}}
\newcommand{\by}{\mbox {\boldmath $y$}}
\newcommand{\bz}{\mbox {\boldmath $z$}}

\newenvironment{Ex}
{\begin{adjustwidth}{0.04\linewidth}{0cm}
\begingroup\small
\vspace{-1.0em}
\raisebox{-.2em}{\rule{\linewidth}{0.3pt}}
\begin{example}
}
{
\end{example}
\vspace{-5mm}
\rule{\linewidth}{0.3pt}
\endgroup
\end{adjustwidth}}
